\documentclass{article}
\usepackage{isabelle,isabellesym}

\usepackage{xcolor}
\usepackage{csquotes}
\usepackage{enumitem}
\usepackage{pdfpages}

\newlist{inlinelist}{enumerate*}{1}
\setlist*[inlinelist,1]{%
  label=(\roman*),
}
\usepackage{hyperref}
\usepackage[numbers]{natbib}
\usepackage{amsmath}
\usepackage{amsthm}
\usepackage{amsfonts}
\usepackage{amssymb}
\usepackage{bbm}  
\usepackage{enumitem}
\usepackage{multidef}
\usepackage{verbatim}
\usepackage{stmaryrd} 

\newcommand{\y}{\wedge}

\DeclareMathOperator{\dom}{dom}

\newcommand{\modelo}[1]{\mathbf{#1}}
\newcommand{\axiomas}[1]{\mathit{#1}}
\newcommand{\clase}[1]{\mathsf{#1}}

\newcommand{\operador}[1]{\mathbf{#1}}

\multidef{\clase{#1}}{Card,HC,HF,Lim,On->Ord,Reg,WF,Ord}

\DeclareMathAlphabet{\mathbbm}{U}{bbm}{m}{n} 
\newcommand{\1}{\mathbbm{1}}
\newcommand{\PP}{\mathbbm{P}}

\multidef[prefix=cal]{\mathcal{#1}}{A-Z}
\multidef{\modelo{#1}}{A,BB->B,CC->C,NN->N,QQ->Q,RR->R,ZZ->Z}

\multidef[prefix=p]{\mathbb{#1}}{A-Z}

\renewcommand{\H}{\operador{H}}
\renewcommand{\S}{\operador{S}}

\renewcommand{\emptyset}{\varnothing}

\newcommand{\Pow}{\mathop{\mathcal{P}}}
\renewcommand{\P}{\Pow}

\newcommand{\lb}{\langle}
\newcommand{\rb}{\rangle}
\newcommand{\impl}{\rightarrow}

\renewcommand{\phi}{\varphi}

\renewcommand{\th}{\theta}

\newcommand{\al}{\alpha}
\newcommand{\be}{\beta}

\newcommand{\defi}{\mathrel{\mathop:}=}
\newcommand{\forces}{\Vdash}

\newcommand{\sig}{\ensuremath{\sigma}}

\newcommand{\CH}{\axiomas{CH}}

\newcommand{\ZF}{\axiomas{ZF}}

\newcommand{\restr}{\upharpoonright}

\newcommand{\union}{\mathop{\textstyle\bigcup}}

\newcommand{\sbq}{\subseteq}

\DeclareMathOperator{\val}{\mathit{val}}
\DeclareMathOperator{\chk}{\mathit{check}}

\renewcommand{\PP}{\mathbb{P}}

\newcommand{\formula}{\isatt{formula}}
\newcommand{\tyi}{\isatt{i}}
\newcommand{\tyo}{\isatt{o}}
\newcommand{\forceisa}{\mathop{\mathtt{forces}}}

\makeatletter
\@ifpackageloaded{txfonts}\@tempswafalse\@tempswatrue
\if@tempswa
  \DeclareFontFamily{U}{txsymbols}{}
  \DeclareFontFamily{U}{txAMSb}{}
  \DeclareSymbolFont{txsymbols}{OMS}{txsy}{m}{n}
  \SetSymbolFont{txsymbols}{bold}{OMS}{txsy}{bx}{n}
  \DeclareFontSubstitution{OMS}{txsy}{m}{n}
  \DeclareSymbolFont{txAMSb}{U}{txsyb}{m}{n}
  \SetSymbolFont{txAMSb}{bold}{U}{txsyb}{bx}{n}
  \DeclareFontSubstitution{U}{txsyb}{m}{n}
  \DeclareMathSymbol{\aleph}{\mathord}{txsymbols}{64}
  \DeclareMathSymbol{\beth}{\mathord}{txAMSb}{105}
  \DeclareMathSymbol{\gimel}{\mathord}{txAMSb}{106}
  \DeclareMathSymbol{\daleth}{\mathord}{txAMSb}{107}
\fi
\makeatother

\newtheorem{theorem}{Theorem}
\newtheorem{lemma}[theorem]{Lemma}

\newtheorem*{claim*}{Claim}
\theoremstyle{definition}
\newtheorem{definition}[theorem]{Definition}

\theoremstyle{remark}
\newtheorem*{remark*}{Remark}

\newcommand{\quantRel}[3]{#1 #2\kern -1pt[#3]}

\newif\ifarXiv
\newif\ifIEEE

\makeatletter
\def\foottext{\gdef\@thefnmark{}\@footnotetext}
\makeatother
\newcommand{\keywords}[1]{\foottext{\emph{Keywords:} #1}}

\hypersetup{
  pdftitle={Mechanization of Separation in Generic Extensions},
  pdfsubject={Computer Science},
  pdfkeywords={Isabelle/ZF, forcing, names, generic extension, constructibility},
  colorlinks,
  linkcolor={blue!40!black},
  citecolor={blue!40!black},
  urlcolor={blue!40!black}
}

\IEEEfalse
\arXivtrue

\begin{document}
\title{Mechanization of Separation in Generic Extensions}
\author{Emmanuel Gunther
  \and 
  Miguel Pagano
  \and 
  Pedro S\'anchez Terraf}
\maketitle

\begin{abstract}
  We mechanize, in the proof assistant
  Isabelle, a proof of the
  axiom-scheme of Separation in 
  generic extensions of models of set theory  
  by using the fundamental theorems of forcing.
  We also formalize the satisfaction of the axioms of
  Extensionality, Foundation, Union, and Powerset. The axiom of
  Infinity is likewise treated, under additional assumptions on the ground
  model.
  In order to achieve these goals, we extended Paulson's library on
  constructibility  with 
  renaming of variables for internalized formulas, improved
  results on 
  definitions by recursion on well-founded  relations, 
  and sharpened hypotheses in his development of relativization and
  absoluteness.
\end{abstract}


\section{Introduction}
Zermelo-Fraenkel Set Theory ($\ZF$) has a prominent place among formal
theories. The reason for this is that it formalizes many intuitive
properties of the notion of set. As such, it can be used as a
foundation for mathematics and thus it has been
thoroughly studied. 
Considering the current trend 
of mechanization of mathematics~\cite{avigad2018mechanization}, it
seems natural to ask for a mechanization of the most salient results
of Set Theory.

The results we are interested in originally arose in connection to
\emph{relative consistency} proofs in set theory; that is, showing
that if $\ZF$ is 
consistent, the addition of a new axiom $A$ won't make the system
inconsistent; this is as much as we can expect to obtain, since
G\"odel's Incompleteness Theorems precludes a  
formal proof of the consistency of set theory in $\ZF$, unless the
latter is indeed inconsistent. There are statements $A$ which are
\emph{undecided} by $\ZF$, in the sense that both $A$ and $\neg A$ are
consistent relative to $\ZF$; perhaps the most prominent  is the
\emph{Continuum Hypothesis}, which led to the development of powerful
techniques for independence proofs. First, G\"odel inaugurated the
theory of \emph{inner models} by introducing his model $L$ of
the \emph{Axiom of Constructibility} \cite{godel-L} and proved the
relative consistency of the Axiom of Choice and the Generalized
Continuum Hypothesis with $\ZF$. More than twenty years later, Paul
J.~Cohen~\cite{Cohen-CH-PNAS} devised the technique of \emph{forcing},
which is the only known way of \emph{extending} models of $\ZF$; this
was used to prove the relative consistency of the 
negation of the Continuum Hypothesis. 

In this work we address a substantial part of formalizing the proof
that given a model $M$ of $\ZF$, any \emph{generic extension} $M[G]$
obtained by forcing also satisfies $\ZF$. As remarked by
\citet[][p.250]{kunen2011set} \enquote{[...] in verifying that $M[G]$
  is a model for set theory, the hardest axiom to verify is
  [Separation].}  
The most important achievement of this paper is the
mechanization in the proof assistant \emph{Isabelle} of a proof of the
Axiom of Separation in 
generic extensions by using the ``fundamental'' theorems of forcing.
En route
to this, we also formalized the satisfaction by $M[G]$ of
Extensionality, Foundation, and Union. As a consequence of Separation
we were able to formalize the proof of the Powerset Axiom; finally,
the Axiom of Infinity was proved under extra assumptions. %
The theoretical support for this work has been 
the fine textbook by Kunen \cite{kunen2011set} and our development
benefited from the remarkable work done by Lawrence 
Paulson \cite{paulson_2003} on the formalization of G\"odel's
constructible universe in Isabelle. 

The
ultimate goal of our project is the formalization of the forcing
techniques needed to show the independence of the Continuum
Hypothesis. 
We think that this project constitutes an interesting test-case
for the current technology of formalization of mathematics, in
particular for the need of handling several layers of reasoning. 

The \emph{Formal Abstracts} project~\cite{hales-fabstracts} proposes
the formalization of complex pieces of mathematics by writing the
statements of 
results and the material upon which they are based (definitions,
propositions, lemmas), but omitting the proofs. In this work we
partially adhere to this vision to delineate our formalization
strategy:
Since the proofs that the  axioms hold in generic extensions
are independent of the \emph{proofs} of the fundamental theorems of
forcing, we assumed the latter for the time being. Let us remark
that those theorems depend on the definition of a function $\forceisa$
from formulas to formulas which is, by itself, quite demanding; the
formalization of it and of the fundamental theorems of forcing %
comprises barely less than a half of our full project.

It might be a little surprising the lack of formalizations of forcing
and generic extensions. As far as we know, the development of
\citet{JFR6232} in homotopy type theory for constructing generic
extensions in a sheaf-theoretic setting is the unique mechanization of
forcing. This contrast with the fruitful use of forcing techniques to
extend the Curry-Howard isomorphism to classical axioms
\cite{Miquel:2011:FPT:2058525.2059614,lmcs:1070}. Moreover, the
combination of forcing with intuitionistic type theory
\cite{Coquand:2009:FTT:1807662.1807665,coquand2010note} gives rise
both to positive results (an algorithm to obtain witnesses of the
continuity of definable functionals \cite{coquand2012computational})
and also negative (the independence of Markov's principle
\cite{lmcs:3859}). In the same strand of forcing from the point of
view of proof theory \cite{avigad_2004} are the conservative
extensions of CoC with forcing conditions
\cite{jaber:hal-01319066,Jaber:2012:ETT:2358958.2359524}.

In pursuing the proof of Separation on generic extensions we
extended Paulson's library with:
\begin{inlinelist}
\item renaming of variables for \emph{internalized} formulas, which
  with little effort can be extended 
  to substitutions;
\item an improvement on definitions by recursion on well-founded
  relations; 
\item enhancements in the hierarchy of locales; and
\item a variant of the  principle of dependent choices and a version
  of Rasiowa-Sikorski, which 
  ensures the existence of generic filters for countable and transitive
  models of $\ZF$;
\end{inlinelist} 
the last item was already communicated in the
  first report \cite{2018arXiv180705174G}.
  
We briefly describe the contents of each
section. Section~\ref{sec:isabelle} contains the bare minimum
requirements to understand the (meta)logics used in Isabelle. Next, an
overview of the model theory of set theory is presented in
Section~\ref{sec:axioms-models-set-theory}. There is an ``internal''
representation of first-order formulas as sets, implemented by
Paulson; Section~\ref{sec:renaming} discusses syntactical
transformations of the former, mainly permutation of variables. 
In Section~\ref{sec:generic-extensions} the generic extensions are
succinctly reviewed and how the treatment of well founded recursion in
Isabelle was enhanced. We take care of the ``easy axioms'' in
Section~\ref{sec:easy-axioms}; these are the ones that
do not depend on the forcing theorems. We describe the latter in
Section~\ref{sec:forcing}. We adapted the  work by Paulson to our
needs, and this is described in
Section~\ref{sec:hack-constructible}. We present the proof
of the Separation Axiom Scheme in Section~\ref{sec:proof-separation},
which follows closely its implementation, and some comments on the
proof of the Powerset Axiom. A plan for future work and
some immediate conclusions are offered in
Section~\ref{sec:conclusions-future-work}.

\section{Isabelle}
\label{sec:isabelle}
\subsection{Logics}
\label{sec:logics}
Isabelle \cite{Isabelle,DBLP:books/sp/Paulson94} provides a
meta-language called \emph{Pure} that consists of a fragment of higher
order logic, where \isatt{\isasymRightarrow} is the function-space
arrow. The meta-Boolean type is called \isatt{prop}. Meta-connectives
\isatt{\isasymLongrightarrow} and \isatt{\&\&\&} fulfill the role of
implication and conjunction, and the meta-binder \isatt{\isasymAnd}
corresponds to universal quantification.

On top of \emph{Pure}, theories/object logics can be defined, with
their own types, connectives and rules. Rules can be written  using
meta-implication: ``$P$, $Q$, and $R$ yield $S$'' can be written
\[
P \ \isatt{\isasymLongrightarrow}\ Q\ \isatt{\isasymLongrightarrow}\ R\ \isatt{\isasymLongrightarrow}\ S
\]
(as usual,  \isatt{\isasymLongrightarrow} associates to the right), and
syntactic sugar is provided to curry the previous rule as follows:
\[
\isasymlbrakk P; Q; R \isasymrbrakk \ \isatt{\isasymLongrightarrow}\ S.
\]
One further example is given by induction on the natural numbers
\isatt{nat},
\[
\isasymlbrakk P(0);\ (\textstyle\isasymAnd
x.\ P(x)\ \isasymLongrightarrow\ P(\isatt{succ}(x))) \isasymrbrakk
\ \isasymLongrightarrow\ P(n), 
\]
where we are omitting the ``typing'' assumtions on $n$ and $x$.

We work in the object theory \emph{Isabelle/ZF}. Two types are defined
in this theory: \tyo, the object-Booleans, and \tyi,
sets. 
It must be observed that predicates (functions with arguments of
type \tyi{} with values in \tyo) do not correspond to first-order formulas;
in particular, those are not recursively constructed.
This will
have concrete consequences in our strategy to approach the
development. From the beginning, we had to resort to
\emph{internalized} formulas \cite[Sect.~6]{paulson_2003}, i.e.\ elements of type $\tyi$ that
encode first-order formulas with a binary relation symbol, and the
satisfaction predicate \isatt{sats\,::\,"i\isasymRightarrow i\isasymRightarrow i\isasymRightarrow o"}  between a set
model with an environment and an internalized formula (where the
relation symbol is interpreted as membership). The set 
\isatt{formula::"\tyi"}
 of internalized
formulas is defined by recursion and hence it is possible to perform
inductive arguments using them. In this sense, the object-logic level
is further divided into \emph{internal} and \emph{external}
sublevels. 

The source code is written for the 2018 version of Isabelle and can be downloaded
from 
\ifIEEE
\begin{center}
\texttt{https://cs.famaf.unc.edu.ar/\~{}mpagano/forcing/}
\end{center}
\fi
\ifarXiv
\url{https://cs.famaf.unc.edu.ar/~mpagano/forcing/}
\fi
(with
minor modifications, it can be run in Isabelle2016-1). Most of it is
presented in the (nowadays standard) declarative flavour called
\emph{Isar} \cite{DBLP:conf/tphol/Wenzel99}, where intermediate
statements in the course of a proof are explicitly stated,
interspersed with automatic tactics handling more trivial steps. The
goal is that the resulting text, a \emph{proof document}, can be
understood without the need of replaying it and viewing the proof state
at each point.

\subsection{Locales}
\label{sec:locales}
Locales \cite{ballarin2010tutorial} provide a neat facility to
encapsulate a context (fixed objects and assumptions on them) that is
to be used in proving several theorems, as in usual mathematical
practice. Furthermore, locales can be organized in hierarchies. 

In this paper, locales have a further use. The \emph{Fundamental
  Theorems of Forcing} we use talk about a specific map $\forceisa$
from formulas to formulas. The definition of $\forceisa$ is involved
and we will not dwell on this now; but applications of those theorems
do not require to know how it is defined. Therefore, we black-box it
and pack everything in a locale called \texttt{forcing\_thms} that
assumes that there is such a 
map that satisfies the Fundamental Theorems.

\section{Axioms and models of set theory}
\label{sec:axioms-models-set-theory}

The axioms of Zermelo and Fraenkel ($\ZF$) form a
countably infinite list of first-order sentences in a language
consisting of an only binary relation symbol $\in$. These include the
axioms of Extensionality, Pairing, Union, Powerset, Foundation, Infinity, and two
axiom-schemes collectively referred as
\begin{inlinelist}[label=(\emph{\alph*})]
\item Axiom of Separation:
\emph{For every $A$, $a_1,\dots,a_n$,  and  a formula
  $\psi(x_0,x_1,\dots,x_n)$, there exists $\{a\in  A:
  \psi(a,a_1,\dots,a_n)\}$}, 
and \item Axiom of Replacement: \emph{For every $A$, $a_1,\dots,a_n$,  and
  a formula   $\psi(x,z,x_1\dots,x_n)$, if 
  $\forall x.\exists!z.\psi(x,z,x_1,\dots,x_n)$,  there exists 
  $\{b : \exists a\in A. \psi(a,b,a_1,\dots,a_n)\}$}.
\end{inlinelist}
An excellent introduction to the axioms and the motivation behind them
can be found in Shoenfield \cite{MR3727410}. 

A model of the theory $\ZF$ consists of a pair $\lb M,E\rb$ where $M$
is a set and $E$ is a binary relation on $M$ satisfying the
axioms. Forcing is a technique 
to extend very special kind of models, where $M$ is a countable
transitive set (i.e., every element of $M$ is a subset of $M$) and
$E$ is the membership relation $\in$ restricted to $M$. In this case
we simply refer to $M$ as a \emph{countable transitive model} or
\emph{ctm}. The following result shows how to obtain ctms from weaker
hypotheses. 
\begin{lemma}\label{lem:wf-model-implies-ctm}
  If there exists  a
  model  $\lb N,E\rb$  of $\ZF$ such that the relation $E$ is well
  founded, then there exists a countable transitive one.
\end{lemma}
\begin{proof}
  (Sketch) The L\"owenheim-Skolem 
  Theorem ensures that there is an countable elementary submodel 
  $\lb N',E\restr N'\rb\preccurlyeq  \lb N,E\rb$ which must also be
  well founded; then the 
  Mostowski collapsing function \cite[Def.~I.9.31]{kunen2011set} sends $\lb
  N',E\restr N'\rb$ 
  isomorphically to some $\lb M,\in\rb$ with  $M$ transitive.
\end{proof}

In this stage of our implementation, we chose a presentation of the
$\ZF$ axioms that would be most
compatible with the development by Paulson. For
instance, the predicate
\isatt{upair{\isacharunderscore}ax{\isacharcolon}{\isacharcolon}{\isachardoublequoteopen}{\isacharparenleft}i{\isacharequal}{\isachargreater}o{\isacharparenright}{\isacharequal}{\isachargreater}o{\isachardoublequoteclose}}
takes a ``class'' (unary predicate) $C$ as an argument and states that
$C$ satisfies the Pairing Axiom.
\begin{isabelle}
upair{\isacharunderscore}ax{\isacharparenleft}C{\isacharparenright}{\isacharequal}{\isacharequal}{\isasymforall}x{\isacharbrackleft}C{\isacharbrackright}{\isachardot}{\isasymforall}y{\isacharbrackleft}C{\isacharbrackright}{\isachardot}{\isasymexists}z{\isacharbrackleft}C{\isacharbrackright}{\isachardot}\ upair{\isacharparenleft}C{\isacharcomma}x{\isacharcomma}y{\isacharcomma}z{\isacharparenright}
\end{isabelle}
Here, $\forall x[C]. \phi$ stands for
$\forall x. C(x) \longrightarrow \phi$, \emph{relative}
quantification. All of the development of relativization by Paulson is
written for a class model, so we set up a locale fixing a set $M$
and using the class $\#\#M\defi \lambda x. \ x\in M$ as the argument. 
\begin{isabelle}
\isacommand{locale}\isamarkupfalse%
\ M{\isacharunderscore}ZF\ {\isacharequal}\ \isanewline
\ \isakeyword{fixes}\ M\ \isanewline
\ \isakeyword{assumes}\ \isanewline
\ \ \ \ \ \ upair{\isacharunderscore}ax{\isacharcolon}\ \ \ \ {\isachardoublequoteopen}upair{\isacharunderscore}ax{\isacharparenleft}{\isacharhash}{\isacharhash}M{\isacharparenright}{\isachardoublequoteclose}\isanewline
\ \ \isakeyword{and}\ \dots \isanewline
\ \ \isakeyword{and}\ separation{\isacharunderscore}ax{\isacharcolon}\isanewline
\ \ {\isachardoublequoteopen}{\isasymlbrakk}\ {\isasymphi}\ {\isasymin}\ formula\ {\isacharsemicolon}\ arity{\isacharparenleft}{\isasymphi}{\isacharparenright}{\isacharequal}{\isadigit{1}}\ {\isasymor}\ arity{\isacharparenleft}{\isasymphi}{\isacharparenright}{\isacharequal}{\isadigit{2}}\ {\isasymrbrakk}\isanewline\ \ \ \ {\isasymLongrightarrow}\isanewline
\ \ {\isacharparenleft}{\isasymforall}a{\isasymin}M{\isachardot}\ separation{\isacharparenleft}{\isacharhash}{\isacharhash}M{\isacharcomma}{\isasymlambda}x{\isachardot}\ sats{\isacharparenleft}M{\isacharcomma}{\isasymphi}{\isacharcomma}{\isacharbrackleft}x{\isacharcomma}a{\isacharbrackright}{\isacharparenright}{\isacharparenright}{\isacharparenright}{\isachardoublequoteclose}\isanewline
\ \ \isakeyword{and}\ replacement{\isacharunderscore}ax{\isacharcolon}\isanewline 
\ \ {\isachardoublequoteopen}{\isasymlbrakk}\ {\isasymphi}\ {\isasymin}\ formula{\isacharsemicolon}\ arity{\isacharparenleft}{\isasymphi}{\isacharparenright}{\isacharequal}{\isadigit{2}}\ {\isasymor}\ arity{\isacharparenleft}{\isasymphi}{\isacharparenright}{\isacharequal}{\isadigit{3}}\ {\isasymrbrakk}\isanewline
\ \ \ \ {\isasymLongrightarrow}\isanewline
\ \ {\isacharparenleft}{\isasymforall}a{\isasymin}M{\isachardot}\ strong{\isacharunderscore}replacement{\isacharparenleft}{\isacharhash}{\isacharhash}M{\isacharcomma}\isanewline
\ \ \ \ {\isasymlambda}x\ y{\isachardot}\ sats{\isacharparenleft}M{\isacharcomma}{\isasymphi}{\isacharcomma}{\isacharbrackleft}x{\isacharcomma}y{\isacharcomma}a{\isacharbrackright}{\isacharparenright}{\isacharparenright}{\isacharparenright}{\isachardoublequoteclose}
\end{isabelle}
The rest of the axioms are also
included. We single out Separation and Replacement: These are written
for formulas with at most one extra parameter (meaning $n\leq 1$ in the
above $\psi$). Thanks to Pairing, these 
versions are equivalent to the usual formulations. We are only able to
prove that the generic extension satisfies Separation for any particular
number of parameters, but not in general. This is a consequence that
induction on terms of type \tyo{} is not available.

It is also possible define a predicate that states that a set
satisfies a (possibly infinite) set of formulas, and then to state
that ``$M$ satisfies $\ZF$'' in a standard way. With the
aforementioned restriction on parameters, it can be shown that this
statement is equivalent to the set of assumptions of the locale
\isatt{M\_ZF}.
\section{Renaming}
\label{sec:renaming}
\newcommand{\renaming}[2]{(#1)[#2]}
\newcommand{\inFm}[2]{#1 \in #2}
\newcommand{\eqFm}[2]{#1 = #2}
\newcommand{\negFm}[1]{\neg #1}
\newcommand{\andFm}[2]{#1 \wedge #2}
\newcommand{\forallFm}[1]{\forall #1}

\newcommand{\inIFm}[2]{\mathsf{Member}(#1,#2)}
\newcommand{\eqIFm}[2]{\mathsf{Equal}(#1,#2)}
\newcommand{\nandIFm}[2]{\mathsf{Nand}(#1,#2)}
\newcommand{\forallIFm}[1]{\mathsf{Forall(#1)}}

In the course of our work we need to reason about renaming of formulas
and its effect on their satisfiability. Internalized formulas are
implemented using de Bruijn indices for variables and the arity of a
formula $\phi$ gives the least natural number containing all the free
variables in $\phi$. Following \citet{fiore-abssyn}, one can
understand the arity of a formula as the context of the free
variables; notice that the arity of $\forallFm{\phi}$ is the
predecessor of the arity of $\phi$. Renamings are, consequently,
mappings between finite sets; since we can think of $\mathsf{succ}(n)$
as the coproduct $1+n = \{0\} \cup \{1,\dots,n\}$, then given a
renaming $f \colon n \to m$, the 
unique morphism $\mathsf{id}_1+f \colon 1+n \to 1+m$ is used to rename
free variables in a quantified formula. 

\begin{definition}[Renaming]
  Let $\phi$ be a formula of arity $n$ and let $f \colon n \to m$, the
  renaming of $\phi$ by $f$, denoted $\renaming{\phi}{f}$, is defined
  by recursion on $\phi$:
  \begin{gather*}
    \renaming{\inFm{i}{j}}{f} = \inFm{f\,i}{f\,j}\\
    \renaming{\eqFm{i}{j}}{f} = \eqFm{f\,i}{f\,j}\\
    \renaming{\negFm{\phi}}{f} = \negFm{\renaming{\phi}{f}}\\
    \renaming{\andFm{\phi}{\psi}}{f} = \andFm{\renaming{\phi}{f}}{\renaming{\psi}{f}}\\
    \renaming{\forallFm{\phi}}{f} = \forallFm{\renaming{\phi}{\mathsf{id}_1+f}}
  \end{gather*}
\end{definition}

As usual, if $M$ is a set, $a_0,\dots,a_{n-1}$ are elements of $M$, and
$\phi$ is a formula of arity $n$, we write
\[
M,[a_0,\dots,a_{n-1}] \models \phi
\]
to denote that $\phi$ is satisfied by $M$ when $i$ is interpreted
as $a_i$ ($i=0,\dots,n-1$). We call the list $[a_0,\dots,a_{n-1}]$ the
\emph{environment}.

The action of renaming on environments re-indexes the variables. An
easy proof connects satisfaction with renamings.
\begin{lemma}
  \label{lem:renaming}
  Let $\phi$ be a formula of arity $n$, $f \colon n \to m$ be a
  renaming, and let $\rho=[a_1,\ldots,a_n]$ and
  $\rho'=[b_1,\ldots,b_m]$ be environments of length $n$ and $m$,
  respectively. If for all $i \in n$, $a_i = b_{j}$ where $j=f\,i$,
  then $M,\rho\models \phi$ is equivalent to
  $M,\rho' \models \renaming{\phi}{f}$.
\end{lemma}

An important resource in Isabelle/ZF is the facility for defining
inductive sets \cite{paulson2000fixedpoint,paulson1995set} together
with a principle for defining functions by structural recursion.
Internalized formulas are a prime example of this, so we define
a function \isa{ren} that associates to each formula an internalized
function that can be later applied to suitable arguments. Notice the
Paulson used \isa{Nand} because it is more economical.
\begin{isabelle}
\isamarkuptrue%
\isacommand{consts}\isamarkupfalse%
\ ren\ {\isacharcolon}{\isacharcolon}\ {\isachardoublequoteopen}i{\isacharequal}{\isachargreater}i{\isachardoublequoteclose}\isanewline
\isacommand{primrec}\isamarkupfalse%
\isanewline
\ {\isachardoublequoteopen}ren{\isacharparenleft}Member{\isacharparenleft}x{\isacharcomma}y{\isacharparenright}{\isacharparenright}\ {\isacharequal}\isanewline
\ \ {\isacharparenleft}{\isasymlambda}\ n\ {\isasymin}\ nat\ {\isachardot}\ {\isasymlambda}\ m\ {\isasymin}\ nat{\isachardot}\ {\isasymlambda}f\ {\isasymin}\ n\ {\isasymrightarrow}\ m{\isachardot}\ Member\ {\isacharparenleft}f{\isacharbackquote}x{\isacharcomma}\ f{\isacharbackquote}y{\isacharparenright}{\isacharparenright}{\isachardoublequoteclose}\isanewline
\ \isanewline
\ {\isachardoublequoteopen}ren{\isacharparenleft}Equal{\isacharparenleft}x{\isacharcomma}y{\isacharparenright}{\isacharparenright}\ {\isacharequal}\isanewline
\ \ {\isacharparenleft}{\isasymlambda}\ n\ {\isasymin}\ nat\ {\isachardot}\ {\isasymlambda}\ m\ {\isasymin}\ nat{\isachardot}\ {\isasymlambda}f\ {\isasymin}\ n\ {\isasymrightarrow}\ m{\isachardot}\ Equal\ {\isacharparenleft}f{\isacharbackquote}x{\isacharcomma}\ f{\isacharbackquote}y{\isacharparenright}{\isacharparenright}{\isachardoublequoteclose}\isanewline
\ \isanewline
\ {\isachardoublequoteopen}ren{\isacharparenleft}Nand{\isacharparenleft}p{\isacharcomma}q{\isacharparenright}{\isacharparenright}\ {\isacharequal}\isanewline
\ \ {\isacharparenleft}{\isasymlambda}\ n\ {\isasymin}\ nat\ {\isachardot}\ {\isasymlambda}\ m\ {\isasymin}\ nat{\isachardot}\ {\isasymlambda}f\ {\isasymin}\ n\ {\isasymrightarrow}\ m{\isachardot}\ \isanewline
\ \ \ Nand\ {\isacharparenleft}ren{\isacharparenleft}p{\isacharparenright}{\isacharbackquote}n{\isacharbackquote}m{\isacharbackquote}f{\isacharcomma}\ ren{\isacharparenleft}q{\isacharparenright}{\isacharbackquote}n{\isacharbackquote}m{\isacharbackquote}f{\isacharparenright}{\isacharparenright}{\isachardoublequoteclose}\isanewline
\ \isanewline
\ {\isachardoublequoteopen}ren{\isacharparenleft}Forall{\isacharparenleft}p{\isacharparenright}{\isacharparenright}\ {\isacharequal}\isanewline
\ \  {\isacharparenleft}{\isasymlambda}\ n\ {\isasymin}\ nat\ {\isachardot}\ {\isasymlambda}\ m\ {\isasymin}\ nat{\isachardot}\ {\isasymlambda}f\ {\isasymin}\ n\ {\isasymrightarrow}\ m{\isachardot}\ \isanewline
\ \ \ Forall\ {\isacharparenleft}ren{\isacharparenleft}p{\isacharparenright}{\isacharbackquote}succ{\isacharparenleft}n{\isacharparenright}{\isacharbackquote}succ{\isacharparenleft}m{\isacharparenright}{\isacharbackquote}sum{\isacharunderscore}id{\isacharparenleft}n{\isacharcomma}f{\isacharparenright}{\isacharparenright}{\isacharparenright}{\isachardoublequoteclose}
\end{isabelle}

In the last equation, \isa{sum{\isacharunderscore}id} corresponds to
the coproduct morphism $\mathsf{id}_{1}+f \colon 1 + n \to 1 +
n$. Since the schema for recursively defined functions does not allow
parameters, we are forced to return a function of three arguments
(\isa{n,m,f}). This also exposes some inconveniences of working in the
untyped realm of set theory; for example to use \isa{ren} we will need
to prove that the renaming is a function. Besides some auxiliary
results (the application of renaming to suitable arguments yields a
formula), the main result corresponding to Lemma~\ref{lem:renaming}
is:
\begin{isabelle}
\isacommand{lemma}\isamarkupfalse%
\ sats{\isacharunderscore}iff{\isacharunderscore}sats{\isacharunderscore}ren\ {\isacharcolon}\ \isanewline
\ \ \isakeyword{fixes}\ {\isasymphi}\isanewline
\ \ \isakeyword{assumes}\ {\isachardoublequoteopen}{\isasymphi}\ {\isasymin}\ formula{\isachardoublequoteclose}\isanewline
\ \ \isakeyword{shows}\ \ {\isachardoublequoteopen}{\isasymAnd}\ n\ m\ {\isasymrho}\ {\isasymrho}{\isacharprime}\ f\ {\isachardot}\ \isanewline
\ \ {\isasymlbrakk}n{\isasymin}nat\ {\isacharsemicolon}\ m{\isasymin}nat\ {\isacharsemicolon}\ f\ {\isasymin}\ n{\isasymrightarrow}m\ {\isacharsemicolon}\ arity{\isacharparenleft}{\isasymphi}{\isacharparenright}\ {\isasymle}\ n\ {\isacharsemicolon}\isanewline
\ \ \ \ \ {\isasymrho}\ {\isasymin}\ list{\isacharparenleft}M{\isacharparenright}\ {\isacharsemicolon}\ {\isasymrho}{\isacharprime}\ {\isasymin}\ list{\isacharparenleft}M{\isacharparenright}\ {\isacharsemicolon}\ \isanewline
\ \ \ {\isasymAnd}\ i\ {\isachardot}\ i{\isacharless}n\ {\isasymLongrightarrow}\ nth{\isacharparenleft}i{\isacharcomma}{\isasymrho}{\isacharparenright}\ {\isacharequal}\ nth{\isacharparenleft}f{\isacharbackquote}i{\isacharcomma}{\isasymrho}{\isacharprime}{\isacharparenright}\ {\isasymrbrakk}\ {\isasymLongrightarrow}\isanewline
\ \ sats{\isacharparenleft}M{\isacharcomma}{\isasymphi}{\isacharcomma}{\isasymrho}{\isacharparenright}\ {\isasymlongleftrightarrow}\ sats{\isacharparenleft}M{\isacharcomma}ren{\isacharparenleft}{\isasymphi}{\isacharparenright}{\isacharbackquote}n{\isacharbackquote}m{\isacharbackquote}f{\isacharcomma}{\isasymrho}{\isacharprime}{\isacharparenright}{\isachardoublequoteclose}\end{isabelle}

All our uses of this lemma involve concrete renamings on small
numbers, but we also tested it with more abstract ones for arbitrary
numbers. All the renamings of the first kind follow the same pattern
and, more importantly, share equal proofs. We would like to develop
some \texttt{ML} tools in order to automatize this.
\section{Generic extensions}
\label{sec:generic-extensions}
We will swiftly review some definitions in order to reach the concept
of \emph{generic extension}. As first preliminary definitions, a \emph{forcing
notion} $\lb\PP,\leq,\1\rb$ is simply a preorder with top, and a \emph{filter}
$G\sbq\PP$ is an increasing subset which is downwards
compatible. Given a ctm $M$ of $\ZF$, a forcing
notion in $M$, and a filter $G$, a new set $M[G]$ is defined. Each
element $a\in M[G]$ is 
determined by its \emph{name} $\dot a$ of $M$. Actually, the structure of
each $\dot a$ is used to construct $a$. They are related by a
map $\val$ that takes $G$ as a parameter:
\[
\val(G,\dot a) = a.
\] 
Then the extension is defined by the image of the map $\val(G,\cdot)$:
\[
M[G] \defi \{\val(G,\tau): \tau\in M\}.
\]
Metatheoretically, it is straightforward to see that $M[G]$ is a
transitive set that satisfies some axioms of $\ZF$ (see
Section~\ref{sec:easy-axioms}) and includes $M\cup\{G\}$. Nevertheless
there is no a priori reason for $M[G]$ to satisfy either Separation, Powerset
or Replacement. The original insight by Cohen was to define the notion
of \emph{genericity} for a filter $G\sbq\PP$ and to prove that
whenever $G$ is generic, $M[G]$ will satisfy $\ZF$. Remember that a
filter is generic if it intersects all the dense sets in $M$; in
\cite{2018arXiv180705174G} we formalized the Rasiowa-Sikorski lemma which
proves the existence of generic filters for ctms.

The Separation Axiom  is the first that requires the notion of
genericity and the use of the forcing machinery, which we review in
the Section~\ref{sec:forcing}.

\subsection{Recursion and values of names}

The map $\val$ used in the definition of the generic extension is
characterized by the recursive equation
\begin{equation}
  \label{eq:val}
  \val(G,\tau) = \{val(G,\sigma) :\exists p \in\PP .%
  \lb\sigma,p\rb \in \tau \wedge p \in G \}
\end{equation}

As is well-known, the principle of  recursion on
well-founded relations \cite[p.~48]{kunen2011set} allows us to define
a recursive function $F \colon A\to A$ by choosing a well-founded
relation $R \subseteq A\times A$ and a functional
$H\colon A\times (A \to A) \to A$ satisfying
$F(a)=H(a,F\!\upharpoonright\!(R^{-1}(a)))$. \citet{paulson1995set}
made this principle available in Isabelle/ZF via the the operator
\isa{wfrec}. The formalization of the corresponding functional
$\mathit{Hv}$ for $\val$ is straightforward:
\begin{isabelle}
\isacommand{definition}\isamarkupfalse%
\isanewline
\ \ Hv\ {\isacharcolon}{\isacharcolon}\ {\isachardoublequoteopen}i{\isasymRightarrow}i{\isasymRightarrow}i{\isasymRightarrow}i{\isachardoublequoteclose}\ \isakeyword{where}\isanewline
\ \ {\isachardoublequoteopen}Hv{\isacharparenleft}G{\isacharcomma}y{\isacharcomma}f{\isacharparenright}\ {\isacharequal}{\isacharequal}\ {\isacharbraceleft}f{\isacharbackquote}x\ {\isachardot}{\isachardot}\ x{\isasymin}domain{\isacharparenleft}y{\isacharparenright}{\isacharcomma}\ {\isasymexists}p{\isasymin}P{\isachardot}\ {\isacharless}x{\isacharcomma}p{\isachargreater}\ {\isasymin}\ y\ {\isasymand}\ p\ {\isasymin}\ G\ {\isacharbraceright}{\isachardoublequoteclose}
\end{isabelle}
In the references \cite{kunen2011set,weaver2014forcing} $\val$ is
applied only to \emph{names}, that are certain elements of $M$
characterized by a recursively defined predicate. The well-founded
relation used to justify Equation~\eqref{eq:val} is
\[ x \mathrel{\mathit{ed}} y \iff \exists p . \lb x,p\rb\in y. \] In
order to use \isa{wfrec} the relation should be expressed as a set, so
in \cite{2018arXiv180705174G} we originally took the restriction of
$\mathit{ed}$ to the whole universe 
$M$; i.e. $\mathit{ed}\cap M\times M$.  Although this decision was
adequate for that work, we now required more flexibility (for
instance, in order to apply $\val$ to arguments that we can't assume
that are in $M$, see Eq.~(\ref{eq:val-of-m}) below).

The remedy is to restrict $\mathit{ed}$ to the
transitive closure of the actual parameter:
\begin{isabelle}
\isacommand{definition}\isamarkupfalse%
\isanewline
\ val\ {\isacharcolon}{\isacharcolon}\ {\isachardoublequoteopen}i{\isasymRightarrow}i{\isasymRightarrow}i{\isachardoublequoteclose}\ \isakeyword{where}\isanewline
\ {\isachardoublequoteopen}val{\isacharparenleft}G{\isacharcomma}{\isasymtau}{\isacharparenright}{\isacharequal}{\isacharequal}\ wfrec{\isacharparenleft}edrel{\isacharparenleft}eclose{\isacharparenleft}{\isacharbraceleft}{\isasymtau}{\isacharbraceright}{\isacharparenright}{\isacharparenright}{\isacharcomma}{\isasymtau}{\isacharcomma}Hv{\isacharparenleft}G{\isacharparenright}{\isacharparenright}{\isachardoublequoteclose}
\end{isabelle}

In order to show that this definition satisfies~(\ref{eq:val}) we had
to supplement the existing recursion tools with a key, albeit
intuitive, result stating that when computing the value of a recursive 
function on some argument $a$, one can restrict the relation to some
ambient set if it includes $a$ and all of its predecessors.
\begin{isabelle}
\isacommand{lemma}\isamarkupfalse%
\ wfrec{\isacharunderscore}restr\ {\isacharcolon}\isanewline
\ \ \isakeyword{assumes}\ {\isachardoublequoteopen}relation{\isacharparenleft}r{\isacharparenright}{\isachardoublequoteclose}\ {\isachardoublequoteopen}wf{\isacharparenleft}r{\isacharparenright}{\isachardoublequoteclose}\ \isanewline
\ \ \isakeyword{shows}\ \ {\isachardoublequoteopen}a{\isasymin}A\ {\isasymLongrightarrow}\ {\isacharparenleft}r{\isacharcircum}{\isacharplus}{\isacharparenright}{\isacharminus}{\isacharbackquote}{\isacharbackquote}{\isacharbraceleft}a{\isacharbraceright}\ {\isasymsubseteq}\ A\ {\isasymLongrightarrow}\ \isanewline
\ \ \ \ \ \ \ \ \ \ wfrec{\isacharparenleft}r{\isacharcomma}a{\isacharcomma}H{\isacharparenright}\ {\isacharequal}\ wfrec{\isacharparenleft}r{\isasyminter}A{\isasymtimes}A{\isacharcomma}a{\isacharcomma}H{\isacharparenright}{\isachardoublequoteclose}
\end{isabelle}
As a consequence, we are able to formalize Equation~(\ref{eq:val}) as follows:
\begin{isabelle}
  \isacommand{lemma}\isamarkupfalse%
  \ def{\isacharunderscore}val{\isacharcolon}\isanewline
  \ {\isachardoublequoteopen}val{\isacharparenleft}G{\isacharcomma}x{\isacharparenright}\ {\isacharequal}\ {\isacharbraceleft}val{\isacharparenleft}G{\isacharcomma}t{\isacharparenright}\ {\isachardot}{\isachardot}\ t{\isasymin}domain{\isacharparenleft}x{\isacharparenright}\ {\isacharcomma}\isanewline
\ \ \ \ \ \ \ \ \ \ \ \ \ \ \ \ \ \ \ {\isasymexists}p{\isasymin}P{\isachardot}\ {\isacharless}t{\isacharcomma}p{\isachargreater}{\isasymin}x\ {\isasymand}\ p{\isasymin}G\ {\isacharbraceright}{\isachardoublequoteclose}
\end{isabelle}
and the monotonicity of $\val$ follows automatically after a
substitution.
\begin{isabelle}
\isacommand{lemma}\isamarkupfalse%
\ val{\isacharunderscore}mono{\isacharcolon}\ {\isachardoublequoteopen}x{\isasymsubseteq}y\ {\isasymLongrightarrow}\ val{\isacharparenleft}G{\isacharcomma}x{\isacharparenright}\ {\isasymsubseteq}\ val{\isacharparenleft}G{\isacharcomma}y{\isacharparenright}{\isachardoublequoteclose}\isanewline
\ \ \isacommand{by}\isamarkupfalse%
\ {\isacharparenleft}subst\ {\isacharparenleft}{\isadigit{1}}\ {\isadigit{2}}{\isacharparenright}\ def{\isacharunderscore}val{\isacharcomma}\ force{\isacharparenright}%
\end{isabelle}
More interestingly we can give a neat equation for values of
names defined by Separation, say $B = \{x\in A\times \PP.\ Q(x)\}$,
then
\begin{equation}
\val(G,B) = \{\val(G,t) : t\in A , \exists p\in \PP \cap G.\ Q(\lb t,p\rb) \} \label{eq:val-name-sep}
\end{equation}

We close our discussion of names and their values by making explicit
the names for elements in $M$; once more, we refer to
\cite{2018arXiv180705174G} for our formalization. The definition of
$\chk(x)$ is a straightforward $\in$-recursion:
\begin{equation*}
  \label{eq:def-check}
  \chk(x)\defi\{\lb\chk(y),\1\rb : y\in x\}
\end{equation*}
An easy $\in$-induction shows $\val(G,\chk(x))=x$.
But to conclude $M\subseteq M[G]$ one also needs to have
$\chk(x) \in M$; this result requires the internalization of
recursively defined functions. This is also needed to prove
$G \in M[G]$; let us define
$\dot G= \{\lb \chk(p),p\rb : p \in \PP \}$, it is easy to prove
$\val(G,\dot G) = G$. Proving $\dot G\in M$ involves knowing
$\chk(x) \in M$ and using one instance of Replacement.

Paulson proved absoluteness results for definitions by recursion and
one of our next goals is to instantiate at $\#\#M$ the appropriate
locale 
\isa{M{\isacharunderscore}eclose} which is the last layer of a pile of
locales. It will take us some time to prove that any ctm of $\ZF$ 
satisfies the
assumptions involved in those locales; as we mentioned, Paulson's work
is mostly done \emph{externally}, i.e. the assumptions are instances
of Separation and Replacement where the predicates and functions are
Isabelle functions of type \isa{i{\isasymRightarrow}i} and
\isa{i{\isasymRightarrow}o}, respectively. In contrast, we assume that
$M$ is a model of $\ZF$, therefore to deduce that $M$ satisfies a
Separation instance, we have to define an internalized formula whose
satisfaction is equivalent to the external predicate (cf. the
interface described in Section~\ref{sec:axioms-models-set-theory} and
also the concrete example given in the proof of Union below).

In the meantime, we declare a locale
\isa{M{\isacharunderscore}extra{\isacharunderscore}assms} assembling
both assumptions ($M$ being closed under $\chk$ and the instance of
Replacement); in this paper we explicitly mention where we use them.

\section{Hacking of \isatt{ZF-Constructible}}
\label{sec:hack-constructible}
In \cite{paulson_2003}, Paulson presented his formalization of the
relative consistency of the Axiom of Choice. This development is
included inside the Isabelle distribution under the session 
\isatt{ZF-Constructible}. The main technical devices, invented by
G\"odel for this purpose, are \emph{relativization} and
\emph{absoluteness}. In a nutshell, to relativize a formula $\phi$ to
a class $C$, it is enough to restrict its quantifiers to $C$. The
example of \isatt{upair\_ax} in
Section~\ref{sec:axioms-models-set-theory}, the relativized version of
the Pairing Axiom, is extracted from \texttt{Relative}, one of the
core theories of \isatt{ZF-Constructible}. On the other hand, $\phi$
is \emph{absolute} for $C$ if it is equivalent to its relativization,
meaning that the statement made by $\phi$ coincides with what $C$
``believes'' $\phi$ is saying. Paulson shows that under certain
hypotheses  on a class $M$ (condensed in the locale \isatt{M\_trivial}), a plethora of
absoluteness and closure results can be proved about $M$.

The development of forcing, and the study of ctms in general, takes
absoluteness as a starting point. We were not able to work with
\isatt{ZF-Constructible} right out-of-the-box. The main reason is that
we can't expect to state the ``class version'' of Replacement for a
\emph{set} $M$ by
using first-order formulas, since predicates
\isatt{P::"i\isasymRightarrow o"} can't
be proved to be only the definable ones. Therefore, we had to make
some modifications in several
locales to make the results available as tools for the present and
future developments.

The most notable changes, located in the theories \texttt{Relative}
and \isatt{WF\_absolute}, are
the following:
\begin{enumerate}
\item\label{item:1} The locale \isatt{M\_trivial}
  does not assume that the underlying class $M$ satisfies the relative
  Axiom of replacement.  As a consequence, the lemma
  \isatt{strong\_replacementI} is no longer valid and was commented
  out.
\item\label{item:2}  Originally the Powerset Axiom was assumed by the
  locale \isatt{M\_trivial},   we moved this requirement to \isatt{M\_basic}. 
\item\label{item:3} We replaced the need that the set of natural
  numbers is in $M$ by the   milder hypothesis that $M(0)$. Actually,
  most results should follow 
  by only assuming that $M$ is non-empty.
\item We moved the requirement $M(\mathtt{nat})$ to the locale
  \isatt{M\_trancl}, where it is needed for the first time. Some results,
  for instance \isatt{rtran\_closure\_mem\_iff} and 
  \isatt{iterates\_imp\_wfrec\_replacement} had to be moved inside that
  locale. 
\end{enumerate}
Because of these changes, some theory files from the
\isatt{ZF-Constructible} session have been included among ours.

The proof, for instance, that the constructible universe $L$ satisfies
the modified locale \isatt{M\_trivial} holds with minor
modifications. Nevertheless, in order to have a neater presentation,
we have stripped off several sections concerning $L$ from the theories
\isatt{L\_axioms} and \isatt{Internalize}, and we merged them to form
the new file  \isatt{Internalizations}. 

\section{Extensionality, Foundation, Union, Infinity}
\label{sec:easy-axioms}

In our first presentation of this project \cite{2018arXiv180705174G},
we proved that $M[G]$ satisfies Pairing; now we have redone this proof
in Isar. It is straightforward to show that the generic extension
$M[G]$ satisfies extensionality and foundation. Showing that it is
closed under Union depends on $G$ being a filter. Infinity is also
easy, but it depends in one further assumption.

For Extensionality in $M[G]$, the assumption 
$\forall w[M[G]]. w\in x \leftrightarrow w\in y$ yields 
$\forall w. w\in x \leftrightarrow w\in y$ by transitivity of $M[G]$. %
Therefore, by (ambient) Extensionality we conclude $x=y$. 

Foundation for $M[G]$ does not depend on $M[G]$ being transitive: in
this case we take $x\in M[G]$ and prove, relativized to $M[G]$,  that there is an
$\in$\kern -1pt-minimal element in $x$. Instantiating the global Foundation
Axiom for $x\cap M[G]$ we get a minimal $y$, so it is still minimal
when considered relative to $M[G]$. 

It is noteworthy that the proofs in the Isar dialect of Isabelle
strictly follow the argumentation of the two previous paragraphs.

The Union Axiom asserts that if $x$ is a set, then there exists
another set (the union of $x$) containing all the elements in each
element of $x$. The relativized version of Union asks to give a name
$\pi_a$ for each $a\in M[G]$ and proving $\val(G,\pi_a)=\union a$.
Let $\tau$ be the name for $a$, i.e.\ $a=\val(G,\tau)$; 
\citet{kunen2011set} gives $\pi_a$ in terms of $\tau$:
\begin{align*}
  \pi_a = \{\langle\theta,p \rangle :  %
\exists \langle\sigma,q\rangle  \in \tau .
 \exists r . \langle \theta,r\rangle \in \sigma \wedge
    p\leqslant r \wedge p \leqslant q \}
\end{align*}
Our formal definition is slightly different in order to ease the proof
of $\pi_a \in M$; as it is defined using Separation, so one needs to
define the domain of separation and also internalize the predicate as
a formula
\isa{union{\isacharunderscore}name{\isacharunderscore}fm}. Instead of
working directly with the internalized formula, we define a predicate
\isa{Union{\isacharunderscore}name{\isacharunderscore}body} and prove the equivalence between
\begin{center}
\isa{sats(M,union{\isacharunderscore}name{\isacharunderscore}fm,[P,leq,\isasymtau,x])}
\end{center}
and
\isa{Union{\isacharunderscore}Name{\isacharunderscore}body(P,leq,\isasymtau,x)}. The
definition of $\pi_a$ in our formalization is:
\begin{isabelle}
\isacommand{definition}\isamarkupfalse%
\ Union{\isacharunderscore}name\ {\isacharcolon}{\isacharcolon}\ {\isachardoublequoteopen}i\ {\isasymRightarrow}\ i{\isachardoublequoteclose}\ \isakeyword{where}\isanewline
\ \ {\isachardoublequoteopen}Union{\isacharunderscore}name{\isacharparenleft}{\isasymtau}{\isacharparenright}\ {\isacharequal}{\isacharequal}\ \isanewline
\ \ \ \ {\isacharbraceleft}u\ {\isasymin}\ domain{\isacharparenleft}{\isasymUnion}{\isacharparenleft}domain{\isacharparenleft}{\isasymtau}{\isacharparenright}{\isacharparenright}{\isacharparenright}\ {\isasymtimes}\ P\ {\isachardot}\isanewline
\ \ \ \  \ \ \ \ Union{\isacharunderscore}name{\isacharunderscore}body{\isacharparenleft}P{\isacharcomma}leq{\isacharcomma}{\isasymtau}{\isacharcomma}u{\isacharparenright}{\isacharbraceright}{\isachardoublequoteclose}
\end{isabelle}

Once we know $\pi_a \in M$, the equation $\val(G,\pi_a)=\union a$ is
proved by showing the mutual inclusion; in both cases one uses that
$G$ is a filter.
\begin{isabelle}
  \isacommand{lemma}\isamarkupfalse%
\ Union{\isacharunderscore}MG{\isacharunderscore}Eq\ {\isacharcolon}\ \isanewline
\ \ \isakeyword{assumes}\ {\isachardoublequoteopen}a\ {\isasymin}\ M{\isacharbrackleft}G{\isacharbrackright}{\isachardoublequoteclose}\ \isakeyword{and}\ {\isachardoublequoteopen}a\ {\isacharequal}\ val{\isacharparenleft}G{\isacharcomma}{\isasymtau}{\isacharparenright}{\isachardoublequoteclose}\ \isakeyword{and}\isanewline
\ \ \ \ \ \ \ \ \ \ {\isachardoublequoteopen}filter{\isacharparenleft}G{\isacharparenright}{\isachardoublequoteclose}\ \isakeyword{and}\ {\isachardoublequoteopen}{\isasymtau}\ {\isasymin}\ M{\isachardoublequoteclose}\isanewline
\ \ \isakeyword{shows}\ {\isachardoublequoteopen}{\isasymUnion}\ a\ {\isacharequal}\ val{\isacharparenleft}G{\isacharcomma}Union{\isacharunderscore}name{\isacharparenleft}{\isasymtau}{\isacharparenright}{\isacharparenright}{\isachardoublequoteclose}
\end{isabelle}
Since Union is absolute for any transitive class we may conclude that
$M[G]$ is closed under Union:
\begin{isabelle}
  \isacommand{lemma}\isamarkupfalse%
  \ union{\isacharunderscore}in{\isacharunderscore}MG\ {\isacharcolon}\ \isanewline
  \ \ \isakeyword{assumes}\ {\isachardoublequoteopen}filter{\isacharparenleft}G{\isacharparenright}{\isachardoublequoteclose}\isanewline
\ \ \isakeyword{shows}\ {\isachardoublequoteopen}Union{\isacharunderscore}ax{\isacharparenleft}{\isacharhash}{\isacharhash}M{\isacharbrackleft}G{\isacharbrackright}{\isacharparenright}{\isachardoublequoteclose}
\end{isabelle}

The proof of Infinity for $M[G]$ takes advantage of some absoluteness
results proved in the locale \isa{M{\isacharunderscore}trivial}; this
proof is easy because we work in the context of the locale
\isa{M{\isacharunderscore}extra{\isacharunderscore}assms} which states
the assumption $\chk(x) \in M$ whenever $x\in M$. Since we have
already proved that $M[G]$ is transitive, $\emptyset\in M[G]$ assuming
$G$ being generic, and also that it satisfies Pairing and Union, we
can instantiate \isa{M{\isacharunderscore}trivial}:
\begin{isabelle}
\isacommand{sublocale}\isamarkupfalse%
\ G{\isacharunderscore}generic\ {\isasymsubseteq}\ M{\isacharunderscore}trivial{\isachardoublequoteopen}{\isacharhash}{\isacharhash}M{\isacharbrackleft}G{\isacharbrackright}{\isachardoublequoteclose}
\end{isabelle}
We assume that $M$ satisfies Infinity, i.e., that Infinity relativized
to $M$ holds; therefore we obtain $I \in M$ such that $\emptyset\in I$
and, $x \in I$ implies $\mathit{succ}(x)\in I$ by absoluteness of
empty and successor for $M$. Using the assumption that $M$ is closed
under $\chk$, we deduce $\val(G,\chk(I)) = I \in M[G]$.  Now we can
use absoluteness of emptiness and successor, this time for $M[G]$, to
conclude that $M[G]$ satisfies Infinity.
\section{Forcing}
\label{sec:forcing}

For the most part, we follow Kunen \cite{kunen2011set}. As
an alternative, introductory 
resource, the  interested reader can check
\cite{chow-beginner-forcing}; the book \cite{weaver2014forcing}
contains a thorough treatment minimizing the technicalities.

Given a ctm $M$, and an $M$-generic filter $G\sbq\PP$, the Forcing
Theorems relate satisfaction of a formula 
$\phi$ in the generic extension $M[G]$ to the satisfaction of another
formula $\forceisa(\phi)$ in $M$. The map $\forceisa$ is defined by
recursion on 
the structure of $\phi$. It is to be noted that the base case (viz.,
for atomic $\phi$) contains all the complexity; the case for
connectives and quantifiers is then straightforward.
In order to state the properties of this map
in sufficient generality to prove that  $M[G]$ satisfies $\ZF$, we work with
internalized formulas, because it is not possible to carry inductive
arguments over \tyo.

We will now make more precise the properties of the map
$\forceisa$ and how it relates satisfaction in $M$ to that in
$M[G]$. Actually, if the formula $\phi$ has $n$ free variables,
$\forceisa(\phi)$ will have $n+4$ free variables, where the first four account
for the forcing notion and a particular element of it. 

We write $\phi(x_0,\dots,x_n)$ to indicate that the free variables of
$\phi$ are among $\{x_0,\dots,x_n\}$. In the case of a formula of the
form $\forceisa(\phi)$, we will make an abuse of notation and indicate
the variables inside the argument of $\forceisa$. As an example, take
the formula $\phi\defi x_1\in x_0$. Then
\[
M,[a,b] \models x_1\in x_0
\]
will hold whenever $b\in a$; and instead of writing $\forceisa(\phi)$
we will write $\forceisa(x_5\in x_4)$, as in
\[
M,[\PP,\leq,\1,p,\tau,\rho] \models \forceisa(x_5\in x_4).
\]

If
$\phi=\phi(x_0,\dots,x_n)$, the notation used by Kunen
\cite{kunen2011set,kunen1980} for $\forceisa(\phi)$ is 
\[
p\forces_{\PP,\leq,\1}^* \phi(x_0,\dots,x_n).
\]
Here, the extra parameters are $\PP,\leq,\1,$ and $p\in\PP$, and the
first three are usually omitted. %
Afterwards, the \emph{forcing relation}
$\forces$ can be obtained by 
interpreting $\forces^*$ in a ctm $M$, for fixed
$\lb\PP,\leq,\1\rb\in M$: $p\forces \phi(\tau_0,\dots,\tau_n)$ holds
if and only if
\begin{equation}\label{eq:3}
M,[\PP,\leq,\1,p,\tau_0,\dots,\tau_n]\models \forceisa(x_4,\dots,x_{n+4}).
\end{equation}

\subsection{The fundamental theorems}
\label{sec:fundamental-theorems}
Modern treatments of the theory of forcing start
by defining the 
forcing relation semantically and later it is proved  that the
characterization given by (\ref{eq:3}) indeed holds, and hence the
forcing relation is \emph{definable}.

Then the definition of the forcing relation is stated as a
\begin{lemma}[Definition of Forcing]\label{lem:definition-of-forcing}
  Let $M$ be a ctm of $\ZF$, $\lb\PP,\leq,\1\rb$ a forcing notion
  in $M$, $p\in\PP$, and $\phi(x_0,\dots,x_n)$ a formula in the
  language of set 
  theory with all free variables displayed. Then the
  following are equivalent, for all $\tau_0,\dots,\tau_n\in M$:
  \begin{enumerate}
  \item $M,[\PP,\leq,\1,p,\tau_0,\dots,\tau_n] 
  \models\forceisa(\phi(x_4,\dots,x_{n+4}))$.
  \item For all $M$-generic filters $G$ such that $p\in G$,
    \ifIEEE
    $M[G],[\val(G,\tau_0),\dots,\val(G,\tau_n)]
    \models\phi(x_0,\dots,x_n)$.
    \fi
    \ifarXiv
    \[
    M[G],[\val(G,\tau_0),\dots,\val(G,\tau_n)]
    \models\phi(x_0,\dots,x_n).
    \]
    \fi
  \end{enumerate}
\end{lemma}

The \emph{Truth Lemma} states that the forcing
relation indeed relates 
satisfaction in $M[G]$ to that in $M$. 
\begin{lemma}[Truth Lemma]\label{lem:truth-lemma}
  Assume the same hypothesis of
  Lemma~\ref{lem:definition-of-forcing}. Then the
  following are equivalent, for all $\tau_0,\dots,\tau_n\in M$, and
  $M$-generic $G$: 
  \begin{enumerate}
  \item $M[G],[\val(G,\tau_0),\dots,\val(G,\tau_n)]
  \models\phi(x_0,\dots,x_n)$.
  \item  There exists $p\in G$ such that 
    \ifIEEE 
    $M,[\PP,\leq,\1,p,\tau_0,\dots,\tau_n] 
    \models\forceisa(\phi(x_4,\dots,x_{n+4}))$.
    \fi
    \ifarXiv 
    \[ 
    M,[\PP,\leq,\1,p,\tau_0,\dots,\tau_n] 
    \models\forceisa(\phi(x_4,\dots,x_{n+4})).
    \]
    \fi
  \end{enumerate}
\end{lemma}
\noindent The previous two results combined are the ones usually called the
\emph{fundamental theorems}. 

The following auxiliary results (adapted from
\cite[IV.2.43]{kunen2011set}) are also handy in forcing arguments.
\begin{lemma}[Strengthening]\label{lem:strengthen} 
  Assume the same hypothesis of Lemma~\ref{lem:definition-of-forcing}.
  $M, [\PP,\leq,\1,p,\dots]\models\forceisa(\phi)$ and $p_1\leq p$
  implies $M, [\PP,\leq,\1,p_1,\dots]\models\forceisa(\phi)$.
\end{lemma}
\begin{lemma}[Density]\label{lem:density}
  Assume the same hypothesis of
  Lemma~\ref{lem:definition-of-forcing}. $M,[\PP,\leq,\1,p,\dots]\models\forceisa(\phi)$   
  if and only if 
  \[
  \{p_1\in \PP : M,[\PP,\leq,\1,p_1,\dots]\models
  \forceisa(\phi)\}
  \]
  is dense below $p$.
\end{lemma}
All these results are proved by recursion in
$\formula$.

The locale \isatt{forcing\_thms} includes all these results as assumptions on the
mapping $\forceisa$, plus a typing condition  and its
effect on arities:
\begin{isabelle}
\isacommand{locale}\isamarkupfalse%
\ forcing{\isacharunderscore}thms\ {\isacharequal}\ forcing{\isacharunderscore}data\ {\isacharplus}\isanewline
\ \isakeyword{fixes}\ forces\ {\isacharcolon}{\isacharcolon}\ {\isachardoublequoteopen}i\ {\isasymRightarrow}\ i{\isachardoublequoteclose}\isanewline
\ \isakeyword{assumes}\
definition{\isacharunderscore}of{\isacharunderscore}forces{\isacharcolon}\isanewline 
\ \ {\isachardoublequoteopen}p{\isasymin}P\ {\isasymLongrightarrow}\ {\isasymphi}{\isasymin}formula\ {\isasymLongrightarrow}\ env{\isasymin}list{\isacharparenleft}M{\isacharparenright}\ {\isasymLongrightarrow}\isanewline
\ \ \ \  sats{\isacharparenleft}M{\isacharcomma}forces{\isacharparenleft}{\isasymphi}{\isacharparenright}{\isacharcomma}\ {\isacharbrackleft}P{\isacharcomma}leq{\isacharcomma}one{\isacharcomma}p{\isacharbrackright}\ {\isacharat}\ env{\isacharparenright}\ {\isasymlongleftrightarrow}\isanewline
\ \ \ \ \ \
{\isacharparenleft}{\isasymforall}G{\isachardot}{\isacharparenleft}M{\isacharunderscore}generic{\isacharparenleft}G{\isacharparenright}{\isasymand}\
p{\isasymin}G{\isacharparenright}\ {\isasymlongrightarrow}\isanewline 
\ \ \ \ \ \ sats{\isacharparenleft}M{\isacharbrackleft}G{\isacharbrackright}{\isacharcomma}{\isasymphi}{\isacharcomma}map{\isacharparenleft}val{\isacharparenleft}G{\isacharparenright}{\isacharcomma}env{\isacharparenright}{\isacharparenright}{\isacharparenright}{\isachardoublequoteclose}\isanewline
\ \ \isakeyword{and}\ \ definability{\isacharbrackleft}TC{\isacharbrackright}{\isacharcolon}\ {\isachardoublequoteopen}{\isasymphi}{\isasymin}formula\ {\isasymLongrightarrow}\isanewline
\ \ \ \ \ \ \ \ \  forces{\isacharparenleft}{\isasymphi}{\isacharparenright}\ {\isasymin}\ formula{\isachardoublequoteclose}\isanewline
\ \ \isakeyword{and}\ \ \ arity{\isacharunderscore}forces{\isacharcolon}\ \ {\isachardoublequoteopen}{\isasymphi}{\isasymin}formula\ {\isasymLongrightarrow}\isanewline
\ \ \ \ \ \ \
arity{\isacharparenleft}forces{\isacharparenleft}{\isasymphi}{\isacharparenright}{\isacharparenright}\
{\isacharequal}\ arity{\isacharparenleft}{\isasymphi}{\isacharparenright}\ {\isacharhash}{\isacharplus}\ {\isadigit{4}}{\isachardoublequoteclose}\isanewline
\ \ \isakeyword{and}\ \dots
\end{isabelle}

The presentation of the Fundamental Theorems of Forcing in a locale
can be regarded as a \emph{formal abstract} as in the project
envisioned by Hales \cite{hales-fabstracts}, where  statements of
mathematical theorems proven in the literature are posed in a language
that is both human- and computer-readable. The point is to take
particular care so that, v.g., there are no missing hypotheses, so it
is possible to take this statement as firm ground on which to start a
formalization of a proof. 

\section{Separation and Powerset}
\label{sec:proof-separation}
We proceed to describe in detail the main goal of this paper, the formalization
of the proof of the Separation Axiom. Afterwards, we sketch  the
implementation of the Powerset Axiom.

This proof of Separation can be found in the file
\verb|Separation_Axiom.thy|. The order chosen to
implement
the proof sought to minimize the cross-reference of facts;  it is
not entirely appropriate for a text version, so we depart from it in
this presentation. Nevertheless, we will refer to each specific block
of code by line number for ease of reference.

The key technical result is the following:
\begin{isabelle}
  \isacommand{lemma}\isamarkupfalse%
  \ Collect{\isacharunderscore}sats{\isacharunderscore}in{\isacharunderscore}MG\ {\isacharcolon}\isanewline
  \ \ \isakeyword{assumes}\isanewline
  \ \ \ \ {\isachardoublequoteopen}{\isasympi}\ {\isasymin}\ M{\isachardoublequoteclose}\ {\isachardoublequoteopen}{\isasymsigma}\ {\isasymin}\ M{\isachardoublequoteclose}\ {\isachardoublequoteopen}val{\isacharparenleft}G{\isacharcomma}\ {\isasympi}{\isacharparenright}\ {\isacharequal}\ c{\isachardoublequoteclose}\isanewline
  \ \ \ \  {\isachardoublequoteopen}val{\isacharparenleft}G{\isacharcomma}\ {\isasymsigma}{\isacharparenright}\ {\isacharequal}\ w{\isachardoublequoteclose}\isanewline
  \ \ \ \ {\isachardoublequoteopen}{\isasymphi}\ {\isasymin}\ formula{\isachardoublequoteclose}\ {\isachardoublequoteopen}arity{\isacharparenleft}{\isasymphi}{\isacharparenright}\ {\isasymle}\ {\isadigit{2}}{\isachardoublequoteclose}\isanewline
  \ \ \isakeyword{shows}\ \ \ \ \isanewline
  \ \ \ \ {\isachardoublequoteopen}{\isacharbraceleft}x{\isasymin}c{\isachardot}\ sats{\isacharparenleft}M{\isacharbrackleft}G{\isacharbrackright}{\isacharcomma}\ {\isasymphi}{\isacharcomma}\ {\isacharbrackleft}x{\isacharcomma}\ w{\isacharbrackright}{\isacharparenright}{\isacharbraceright}{\isasymin}\ M{\isacharbrackleft}G{\isacharbrackright}{\isachardoublequoteclose}
\end{isabelle}
From this, using absoluteness, we will be able to derive the
$\phi$-instance of Separation. 

To show that   
\[
S\defi\{x\in c : M[G],[x,w]\models \phi(x_0,x_1)\} \in M[G],
\]
it is enough to provide a name $n\in M$ for this set.
 
The candidate name is
\begin{equation}\label{eq:4}
n \defi \{u \in\dom(\pi)\times\PP :M,[u,\PP,\leq,\1,\sig,\pi]\models \psi\}
\end{equation}
where
\[
\psi \defi \exists \th\, p.\ x_0=\lb\th,p\rb \y 
   \forceisa(\th\in x_5\y\phi(\th,x_4)).
\]
The fact that $n\in M$ follows (lines 216--220 of the source file) by
an application of a six-variable instance of Separation in $M$ (lemma
\isatt{six{\isacharunderscore}sep{\isacharunderscore}aux}). We note in
passing that it is possible to  abbreviate expressions in Isabelle by
the use of \textbf{let} statements or  \textbf{is} qualifiers,
and metavariables (whose
identifiers start with a question mark). In this way, the definition in
(\ref{eq:4}) appears in the sources as letting \isatt{?n} to be that
set (lines 208--211).

Almost a third part of the proof involves the syntactic handling of
internalized formulas and permutation of variables. The more
substantive portion concerns proving that actually $\val(G,n)=S$.

Let's first focus into the predicate 
\begin{equation}\label{eq:1}
M,[u,\PP,\leq,\1,\sig,\pi]\models \psi
\end{equation}
defining $n$ by separation. By definition of the satisfaction
relation and %
absoluteness, we have (lines 92--98) that it is equivalent to the fact
that there exist $\th,p\in M$ with   $u=\lb\th,p\rb$  and 
\[
M,[\PP,\leq,\1,p,\th,\sig,\pi]\models \forceisa(x_4\in
x_6\y\phi(x_4,x_5)). 
\]
This, in turn, is equivalent by the Definition of Forcing to: \emph{For all $M$-generic
filters $F$ such that $p\in F$,} 
\begin{equation}\label{eq:2}
M[F],[\val(F,\th),\val(F,\sig),\val(F,\pi)]\models x_0\in
x_2\y\phi(x_0,x_1). 
\end{equation}
(lines 99--185). We can instantiate this statement with $G$ and obtain
(lines 186--206)
\[
p\in G \impl M[G],[\val(G,\th),w,c]\models x_0\in
x_2\y\phi(x_0,x_1). 
\] 
Let $Q(\th,p)$ be the last displayed statement. We have just seen that
(\ref{eq:1}) implies 
\[
\exists \th,p\in M.\ u=\lb\th,p\rb \y Q(\th,p).
\]
Hence (lines 207-212) $n$ is included in 
\[
m\defi \{u \in\dom(\pi)\times\PP : \exists \th,p\in M.\ u=\lb\th,p\rb
\y Q(\th,p)\}. 
\]

Since $m$ is a name defined using Separation, we may use
(\ref{eq:val-name-sep}) to show (lines 221--274 of
\isatt{Separation\_Axiom})
\begin{equation}
  \label{eq:val-of-m}
  \val(G,m) = \{x\in c : M[G],[x,w,c]\models \phi(x_0,x_1)\}.
\end{equation}
The right-hand side is trivially equal to $S$, but as a consequence of
the definition of 
\isatt{separation\_ax}, the result contains an extra $c$ in the
environment.

Also, by monotonicity of $\val$ we obtain 
  $\val(G,n)\sbq \val(G,m)$ (lines 213--215).
To complete the proof, it is therefore enough to show the other
inclusion (starting at line 275).
For this, let $x\in \val(G,m) = S$ and then $x\in c$. Hence there exists
$\lb\th,q\rb\in\pi$ such that $q\in G$ and $x=\val(G,\th)$. 

On the other hand, since (line 297)
\[
M[G],[\val(G,\th),\val(G,\sig),\val(G,\pi)]\models
 x_0\in x_2\y\phi( x_0, x_1),
\]
by the  Truth Lemma there must exist $r\in G$ such that
\[
M,[\PP,\leq,\1,r,\th,\sig,\pi]\models
\forceisa(x_4\in x_6\y\phi( x_4, x_5)).
\]
Since $G$ is a filter, there is $p\in G$ such that $p\leq q, r$.
By Strengthening, we have
\[
M,[\PP,\leq,\1,p,\th,\sig,\pi]\models
\forceisa(x_4\in x_6\y \phi( x_4, x_5)),
\]
which by the Definition of Forcing gives us (lines 315--318): \emph{for all $M$-generic $F$,
  $p\in F$ implies} 
\[
M[F],[\val(F,\th),\val(F,\sig),\val(F,\pi)]\models
 x_0\in  x_2 \y\phi( x_0, x_1).
\]
Note this is the same as (\ref{eq:2}). Hence, tracing the equivalence
up to (\ref{eq:1}), we can show that $x=\val(G,\th)\in \val(G,n)$
(lines 319--337), finishing the main lemma.

The last 20 lines of the theory show, using absoluteness, the two
instances of Separation for $M[G]$:

\begin{isabelle}
\isacommand{theorem}\isamarkupfalse%
\ separation{\isacharunderscore}in{\isacharunderscore}MG{\isacharcolon}\isanewline
\ \ \isakeyword{assumes}\ \isanewline
\ \ \ \ {\isachardoublequoteopen}{\isasymphi}{\isasymin}formula{\isachardoublequoteclose}\ \isakeyword{and}\isanewline
\ \ \ \  {\isachardoublequoteopen}arity{\isacharparenleft}{\isasymphi}{\isacharparenright}\ {\isacharequal}\ {\isadigit{1}}\ {\isasymor}\ arity{\isacharparenleft}{\isasymphi}{\isacharparenright}{\isacharequal}{\isadigit{2}}{\isachardoublequoteclose}\isanewline
\ \ \isakeyword{shows}\ \ \isanewline
\ \ \ \ {\isachardoublequoteopen}{\isasymforall}a{\isasymin}{\isacharparenleft}M{\isacharbrackleft}G{\isacharbrackright}{\isacharparenright}{\isachardot}\isanewline 
 \ \ \ \ \  separation{\isacharparenleft}{\isacharhash}{\isacharhash}M{\isacharbrackleft}G{\isacharbrackright}{\isacharcomma}{\isasymlambda}x{\isachardot}sats{\isacharparenleft}M{\isacharbrackleft}G{\isacharbrackright}{\isacharcomma}{\isasymphi}{\isacharcomma}{\isacharbrackleft}x{\isacharcomma}a{\isacharbrackright}{\isacharparenright}{\isacharparenright}{\isachardoublequoteclose}
\end{isabelle}
   
We now turn to the Powerset Axiom. We followed the proof of
\cite[IV.2.27]{kunen2011set}, to which we refer the reader for further
details. Actually, the main technical result,
\begin{isabelle}
\isacommand{lemma}\isamarkupfalse%
\ Pow{\isacharunderscore}inter{\isacharunderscore}MG{\isacharcolon}\isanewline
\ \ \isakeyword{assumes}\isanewline
\ \ \ \ {\isachardoublequoteopen}a{\isasymin}M{\isacharbrackleft}G{\isacharbrackright}{\isachardoublequoteclose}\isanewline
\ \ \isakeyword{shows}\isanewline
\ \ \ \ {\isachardoublequoteopen}Pow{\isacharparenleft}a{\isacharparenright}\ {\isasyminter}\ M{\isacharbrackleft}G{\isacharbrackright}\ {\isasymin}\ M{\isacharbrackleft}G{\isacharbrackright}{\isachardoublequoteclose}
\end{isabelle}
keeps most of the structure of the printed proof; this ``skeleton'' of the
argument takes around 120 (short) lines, where we tried to  preserve the
names of variables used in the textbook (with the occasional question
mark that distinguishes meta-variables). There are approximately 30
more lines of bureaucracy in the proof of the last lemma. 

Two more
absoluteness lemmas concerning powersets were needed: These are
refinements of results (\isatt{powerset\_Pow} and
\isatt{powerset\_imp\_subset\_Pow}) located in the theory
\isatt{Relative} 
where we weakened the assumption ``$y\in M$'' (\isatt{M(y)}) to 
``$y\sbq  M$'' (second assumption below).  
\begin{isabelle}
\isacommand{lemma}\isamarkupfalse%
\ {\isacharparenleft}\isakeyword{in}\ M{\isacharunderscore}trivial{\isacharparenright}\ powerset{\isacharunderscore}subset{\isacharunderscore}Pow{\isacharcolon}\isanewline
\ \ \isakeyword{assumes}\ \isanewline
\ \ \ \ {\isachardoublequoteopen}powerset{\isacharparenleft}M{\isacharcomma}x{\isacharcomma}y{\isacharparenright}{\isachardoublequoteclose}\ {\isachardoublequoteopen}{\isasymAnd}z{\isachardot}\ z{\isasymin}y\ {\isasymLongrightarrow}\ M{\isacharparenleft}z{\isacharparenright}{\isachardoublequoteclose}\isanewline
\ \ \isakeyword{shows}\ \isanewline
\ \ \ \ {\isachardoublequoteopen}y\ {\isasymsubseteq}\ Pow{\isacharparenleft}x{\isacharparenright}{\isachardoublequoteclose}
\end{isabelle}
\begin{isabelle}
\isacommand{lemma}\isamarkupfalse%
\ {\isacharparenleft}\isakeyword{in}\ M{\isacharunderscore}trivial{\isacharparenright}\ powerset{\isacharunderscore}abs{\isacharcolon}\ \isanewline
\ \ \isakeyword{assumes}\isanewline
\ \ \ \ {\isachardoublequoteopen}M{\isacharparenleft}x{\isacharparenright}{\isachardoublequoteclose}\ {\isachardoublequoteopen}{\isasymAnd}z{\isachardot}\ z{\isasymin}y\ {\isasymLongrightarrow}\ M{\isacharparenleft}z{\isacharparenright}{\isachardoublequoteclose}\isanewline
\ \ \isakeyword{shows}\isanewline
\ \ {\isachardoublequoteopen}powerset{\isacharparenleft}M{\isacharcomma}x{\isacharcomma}y{\isacharparenright}\ {\isasymlongleftrightarrow}\ y\ {\isacharequal}\ {\isacharbraceleft}a{\isasymin}Pow{\isacharparenleft}x{\isacharparenright}\ {\isachardot}\ M{\isacharparenleft}a{\isacharparenright}{\isacharbraceright}{\isachardoublequoteclose}
\end{isabelle}

Of the rest of the theory file \verb|Powerset_Axiom.thy|, a
considerable fraction is taken by the proof of a closure property of
the ctm $M$, that involves renaming of an internalized formula; also,
the handling of the projections \isatt{fst} and \isatt{snd} must be
done internally. 
\begin{isabelle}
\isacommand{lemma}\isamarkupfalse%
\ sats{\isacharunderscore}fst{\isacharunderscore}snd{\isacharunderscore}in{\isacharunderscore}M{\isacharcolon}\isanewline
\ \isakeyword{assumes}\ \isanewline
\  \ {\isachardoublequoteopen}A{\isasymin}M{\isachardoublequoteclose}\ {\isachardoublequoteopen}B{\isasymin}M{\isachardoublequoteclose}\ {\isachardoublequoteopen}{\isasymphi}\ {\isasymin}\ formula{\isachardoublequoteclose}\ {\isachardoublequoteopen}p{\isasymin}M{\isachardoublequoteclose}\ {\isachardoublequoteopen}l{\isasymin}M{\isachardoublequoteclose}\isanewline
\ \  {\isachardoublequoteopen}o{\isasymin}M{\isachardoublequoteclose}\ {\isachardoublequoteopen}{\isasymchi}{\isasymin}M{\isachardoublequoteclose}\ {\isachardoublequoteopen}arity{\isacharparenleft}{\isasymphi}{\isacharparenright}\ {\isasymle}\ {\isadigit{6}}{\isachardoublequoteclose}\isanewline
\  \isakeyword{shows}\isanewline
\  \ {\isachardoublequoteopen}{\isacharbraceleft}sq\ {\isasymin}\ A{\isasymtimes}B\ {\isachardot}\isanewline
\ \ \ \ \ sats{\isacharparenleft}M{\isacharcomma}{\isasymphi}{\isacharcomma}{\isacharbrackleft}p{\isacharcomma}l{\isacharcomma}o{\isacharcomma}snd{\isacharparenleft}sq{\isacharparenright}{\isacharcomma}fst{\isacharparenleft}sq{\isacharparenright}{\isacharcomma}{\isasymchi}{\isacharbrackright}{\isacharparenright}{\isacharbraceright}\ {\isasymin}\ M{\isachardoublequoteclose}
\end{isabelle}

\section{Conclusions and future work}
\label{sec:conclusions-future-work}

The ultimate goal of our project is a complete mechanization of
forcing allowing for further developments (formalization of the
relative consistency of $\CH$), with the long-term hope that 
working set-theorists will adopt these formal tools as an aid to their
research. In the current paper we reported a first major
milestone towards that goal; viz. a formal proof in Isabelle/ZF of the
satisfaction by generic extensions of most of the $\ZF$ axioms.%

We cannot overstate the importance of following the sharp and detailed 
presentation of forcing given by
\citet{kunen2011set}. In fact, it helped us to delineate the
\emph{thematic} aspects of our formalization; i.e.~the handling of all
the theoretical concepts and results in the subject and it informed
the structure of locales organizing our development. This had an
impact, in particular, in the formal statement of the Fundamental
Theorems. We consider that the writing of the \isatt{forcing\_thms}
locale, though only taking a few lines of code, is the second most
important achievement of this work, since there is no obvious
reference from which to translate this directly. The accomplishment
of the formalizations of Separation and Powerset are, in a sense,
certificates that the locale of the Fundamental Theorems was set
correctly.

Two axioms have not been treated in full. Infinity was proved under
two extra assumptions on the model; when we develop a full-fledged
interface between ctms of $\ZF$  and the locales 
providing recursive constructions from Paulson's
\isatt{ZF-Constructible} session, the same current proof will hold
with no extra assumptions. The same goes for the results $M\sbq M[G]$
and $G\in M[G]$. 

The Replacement Axiom, however, requires some more work to be
done. In Kunen it
requires a relativized version (i.e., showing that it holds for $M$)
of the \emph{Reflection Principle}. In order to state this
meta-theoretic result by Montague, recall that an equivalent
formulation of the Foundation Axiom states that the universe of sets
can be decomposed in a transfinite, cumulative hierarchy of sets:
\begin{theorem}
  Let $V_{\al}\defi\union\{\P(V_\be) : \be<\al\}$ for each ordinal
  $\al$. Then each $V_\al$ is a set and 
  $\forall x. \exists\al .\ \Ord(\al) \y x\in V_\al$.  
\end{theorem}
\begin{theorem}[Reflection Principle]\label{th:reflection-principle}
  For every finite $\Phi\sbq\ZF$, $\ZF$ proves: ``There exist
  unboundedly many $\al$ such that $V_\al\models \Phi$.''
\end{theorem}
It is obvious that we can take the conjunction of $\Phi$ and state
Theorem~\ref{th:reflection-principle} for a single formula, say $\phi$.
The schematic nature of this result
hints at a proof by induction on formulas, and hence it must be shown
internally. It is to be noted that Paulson
\cite{DBLP:conf/cade/Paulson02} also formalized the
Reflection principle in Isabelle/ZF, but it is not clear if the
relativized version follows directly from it. (It may be possible to
sidestep Reflection, since in Neeman \cite{neeman-course}, only the
relativization of the cumulative hierarchy is needed; nevertheless, it
is a nontrivial task.)

This is an appropriate point to insist that the internal/external
dichotomy has been a powerful agent in the shaping of our project.
This tension was also pondered by Paulson in his formalization of
G\"odel's constructible universe \cite{paulson_2003}; after choosing a
shallow embedding of $\ZF$, every argument proved by induction on
formulas (or functions defined by recursion) should be done using
internalized formulas. Working on top of Paulson's library, we
prototyped the formula-transformer $\forceisa$, which is defined for
internalized formulas, and this affects indirectly the proof of the
Separation Axiom (despite the latter is not by induction). The proof
of Replacement also calls for internalized formulas, because one needs
a general version of the Reflection Principle (since the formula
$\phi$ involved depends on the specific instance of Replacement being
proved). 

An alternative road to internalization would be to redevelop
absoluteness 
results in a more structured metatheory that already includes a
recursively defined type of first order formulas. Needless to say, this
comprises an extensive re-engineering.   

A secondary, more prosaic, outcome of this project is to precisely
assess which assumptions on the ground model $M$ are needed to develop
the forcing machinery. The obvious are transitivity and $M$ being
countable (but keep in mind Lemma~\ref{lem:wf-model-implies-ctm}); the
first because many absoluteness results follows from this, the latter
for the existence of generic filters. A
more anecdotal one is that to show that an instance of Separation
with at most two parameters holds in $M[G]$, one needs to assume a
particular  six-parameter
instance in $M$ (four extra parameters can be directly blamed on
$\forceisa$).  The purpose of identifying those assumptions is to assemble
in a locale the specific (instances of) axioms that should satisfy the
ground model in order to perform forcing constructions; this list will
likely include all the instances of Separation and Replacement that
are needed to satisfy the requirements of the locales in the
\isatt{ZF-Constructible} session.

We have already commented on our hacking of \isatt{ZF-Constructible}
to maximize its modularity and thus the re-usability in other
formalizations. We think it would be desirable to organize it
somewhat differently: a trivial change is to catalog in one file all
the internalized formulas. A more
conceptual modification would be to start out with an even more basic
locale that only assumes $M$ to be a non-empty transitive class, as
many absoluteness results follow from this hypothesis. Furthermore, as
Paulson comments in the 
sources, it would have been better to minimize the use of the Powerset
Axiom in locales and proofs. There are useful natural models that
satisfy a sub-theory of $\ZF$ not including Powerset, and to ensure a
broader applicability, it would be convenient to have absoluteness
results not assuming it. We plan to contribute back to the official
distribution of Isabelle/ZF with a thorough revision of the development of
constructibility.

\
\providecommand{\noopsort}[1]{}
\begin{small}\end{small}

\newpage
\appendix
\section*{A short overview of our development}
In this appendix we succinctly describe the contents of each file. We
include in 
Figure~\ref{fig:deps}  a dependency graph of our formalization. The
theories on a grayish background are directly from Paulson; we
highlight with blue/cyan those of Paulson that we modified. We
have developed from scratch the rest, in white.

\begin{description}[itemsep=1.5pt]
\item[\texttt{Nat\_Miscellanea}]Miscellaneous results for naturals, mostly
  needed for renamings.
\item[\texttt{Renaming}] Renaming of internalized formulas, see
  Section \ref{sec:renaming}.
\item[\texttt{Pointed\_DC}] A pointed version of the Principle of
  Dependent Choices.
\item[\texttt{Recursion\_Thms}] Enhancements about recursively defined
  functions.
\item[\texttt{Forcing\_Notions}] Definition of Posets with maximal
  element, filters, dense sets. Proof of the Rasiowa-Sikorski Lemma.
\item[\texttt{Forcing\_Data}] Definition of the locales:
  \begin{inlinelist}
  \item \texttt{M\_ZF} satisfaction of axioms; and 
  \item \texttt{forcing\_data} extending the previous one with
    \texttt{forcing\_notion}, transitivity, and being countable.
  \end{inlinelist}
\item[\texttt{Interface}] Instantiation of locales \texttt{M\_trivial}
  and \texttt{M\_basic} for every instance of \texttt{Forcing\_Data}.
\item[\texttt{Names}] Definitions of $\chk$, $\val$, and the generic
  extension. Various results about them.
\item[\texttt{Forcing\_Theorems}] Specification of fundamental
  theorems of forcing, see Section~\ref{sec:forcing}.
\item[\texttt{*\_Axiom}] Proof of the satisfaction of the
  corresponding axiom in the generic extension.
\end{description}
\begin{figure*}[!b]
  \begin{center}
    \includegraphics[height=14cm]{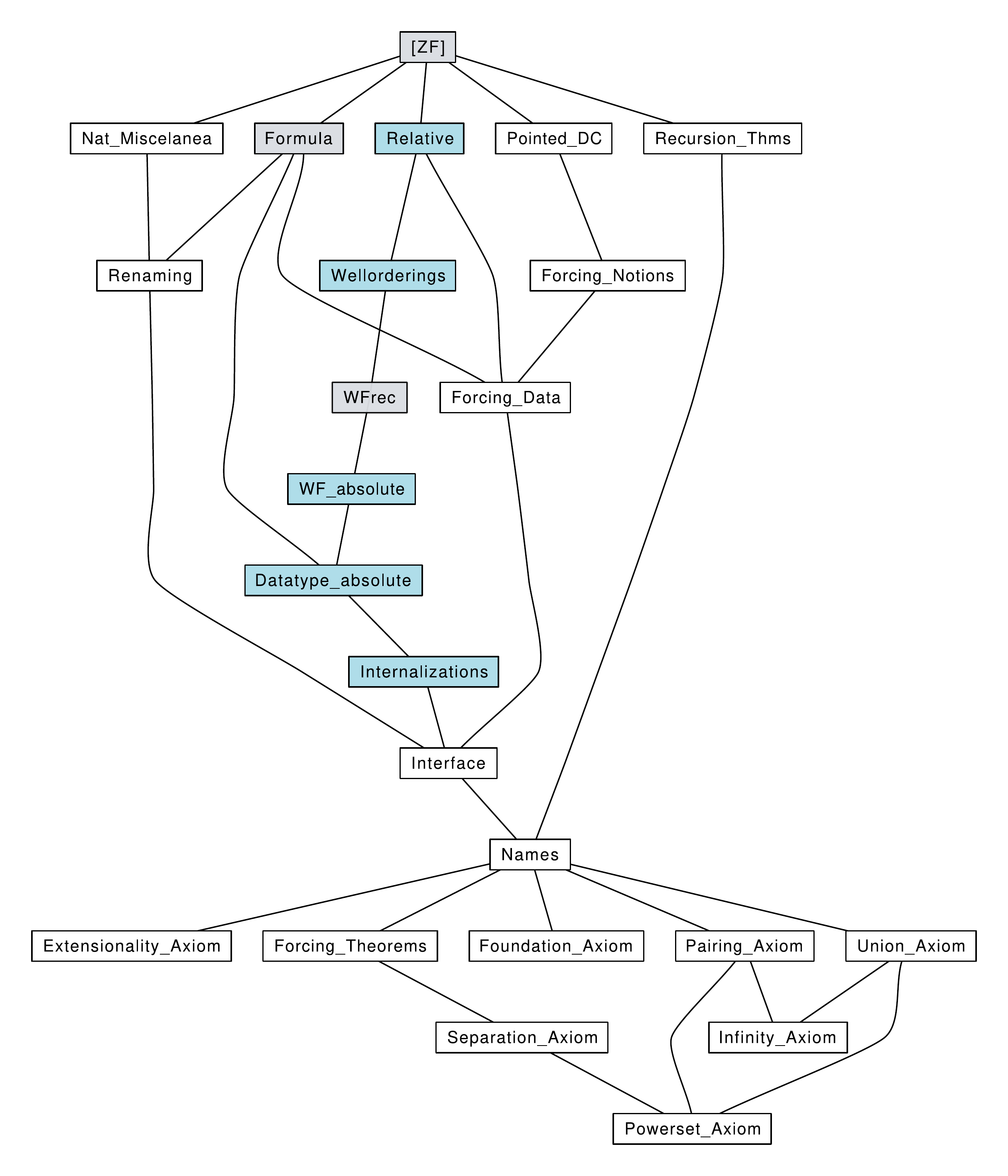}
  \end{center}
  \caption{Dependency graph of the \isatt{Separation} session.}
  \label{fig:deps}
\end{figure*}

\end{document}